\documentclass[aps,preprint,tightenlines,unsortedaddress,11pt]{revtex4}%
\usepackage{amsfonts}
\usepackage{amsmath}
\usepackage{amssymb}
\usepackage{graphicx}%
\providecommand{\U}[1]{\protect\rule{.1in}{.1in}}
\newtheorem{theorem}{Theorem}

\newtheorem{lemma}[theorem]{Lemma}

\newenvironment{proof}[1][Proof]{\noindent\textbf{#1.} }{\ \rule{0.5em}{0.5em}}
\begin{document}
\preprint{UMTG-10}
\title[GBI]{Multi-operator brackets acting thrice}
\author{Thomas Curtright, Xiang Jin, and Luca Mezincescu}
\affiliation{Department of Physics, University of Miami, Coral Gables, FL 33124-8046,
USA\medskip}
\pacs{PACS number}

\begin{abstract}
\medskip

We generalize an identity first found by Bremner for Nambu 3-brackets. \ For
odd $N$-brackets built from associative operator products, we show that
\[
\left[  \left[  A\left[  B_{1}\cdots B_{N}\right]  B_{N+1}\cdots
B_{2N-2}\right]  B_{2N-1}\cdots B_{3N-3}\right]  =\left[  \left[  AB_{1}\cdots
B_{N-1}\right]  \left[  B_{N}\cdots B_{2N-1}\right]  B_{2N}\cdots
B_{3N-3}\right]
\]
for any fixed $A$, when totally antisymmetrized over all the $B$s.

\end{abstract}
\volumeyear{year}
\volumenumber{number}
\issuenumber{number}
\eid{identifier}
\startpage{1}
\endpage{ }
\maketitle

\section{Introduction}

Nambu introduced a multilinear operator bracket in the context of a novel
formulation of mechanics \cite{Nambu}. \ His $N$-bracket is defined by
\begin{equation}
\left[  A_{1}A_{2}\cdots A_{N}\right]  =\sum_{\sigma\in S_{N}}%
\operatorname*{sgn}\left(  \sigma\right)  ~A_{\sigma_{1}}\cdots A_{\sigma_{N}%
}\ ,
\end{equation}
where the sum is over all $N!$ permutations of the operators. \ For example,
$\left[  ABC\right]  =ABC-ACB+BCA-BAC+CAB-CBA$. \ The operator product is
assumed to be associative. \ To avoid ambiguities when some of the entries
within a bracket are themselves products, commas are often used to separate
the entries. \ Parentheses also suffice in such cases. \ For example, $\left[
AD,B,C\right]  =\left[  \left(  AD\right)  BC\right]
=ADBC-ADCB+BCAD-BADC+CADB-CBAD$.

The same construction independently appeared in the mathematical literature
\cite{Higgins,Kurosh}. \ The theory of such multi-operator products, as well
as their \textquotedblleft classical limits\textquotedblright\ in terms of
multivariable Jacobians, has been studied extensively
\cite{Bremner,CFJMZ,Curtright,deAzcarraga,Dito,Filippov,Gautheron,Hanlon,Lada,Pojidaev,Schlesinger,Takhtajan,Vainerman,Vaisman}%
. \ 

From an algebraic point of view, it is natural to seek the analogue of the
Jacobi identity for $N$-brackets. \ For the case of even $N$-brackets, the
obvious generalization where one $N$-bracket acts on another leads to a true
identity \emph{if} all entries are totally antisymmetrized (see (\ref{EvenGJI}%
) below). \ But for odd $N$-brackets this procedure does not work
\cite{Bremner,deAzcarraga} -- the total antisymmetrization over all entries of
one odd $N$-bracket acting on another does not vanish, but rather yields a
higher-order $\left(  2N-1\right)  $-bracket. \ 

Nevertheless, an interesting generalization of the Jacobi identity was
discovered by Bremner for $3$-brackets acting thrice \cite{Bremner}. \ He
showed
\begin{equation}
\left[  \left[  A\left[  bcd\right]  e\right]  fg\right]  =\left[  \left[
Abc\right]  \left[  def\right]  g\right]  \ , \label{B}%
\end{equation}
where $A$ is fixed, but it is implicitly understood that lower case entries
are totally antisymmetrized by summing over all $6!$ signed permutations of
them. \ The point of this short paper is to show that Bremner's identity
generalizes to all odd $N$-brackets. \ 

Before discussing the general case, we anticipate and indicate a proof for the
case of $3$-brackets. \ The Bremner identity can be proved through a
resolution of both left- and right-hand sides as a series of canonically
ordered words. \ By direct calculation we find%
\begin{align}
\left[  \left[  A\left[  bcd\right]  e\right]  fg\right]   &
=24~Abcdefg-36~bAcdefg+36~bcAdefg-24~bcdAefg\nonumber\\
&  +36~bcdeAfg-36~bcdefAg+24~bcdefgA\ ,
\end{align}
where all lower case entries are implicitly totally antisymmetrized.
\ Precisely the same expansion holds for $\left[  \left[  Abc\right]  \left[
def\right]  g\right]  $, again by direct calculation. \ Hence the identity is
established. \ 

That is to say, both $\left[  \left[  A\left[  bcd\right]  e\right]
fg\right]  $ and $\left[  \left[  Abc\right]  \left[  def\right]  g\right]  $
can be rendered as a 7-bracket plus another 3-bracket containing 3-brackets,
when antisymmetrized over lower case entries.%
\begin{equation}
\left[  \left[  A\left[  bcd\right]  e\right]  fg\right]  =\frac{1}%
{20}~\left[  Abcdefg\right]  -\frac{1}{6}~\left[  A\left[  bcd\right]  \left[
efg\right]  \right]  =\left[  \left[  Abc\right]  \left[  def\right]
g\right]  \ .
\end{equation}
Thus the Bremner identity amounts to the combinatorial statement, as written,
that there are two distinct ways to present a 7-bracket in terms of nested 3-brackets.

\section{Results for any N}

As known, and previously mentioned, even brackets need only act twice to yield
an identity. \ Namely \cite{deAzcarraga,Hanlon},
\begin{equation}
\left[  B_{1}\cdots B_{N-1}\left[  B_{N}\cdots B_{2N-1}\right]  \right]
=0\text{ \ \ for }N\text{ even.} \label{EvenGJI}%
\end{equation}
Total antisymmetrization of all the $B$s is understood \cite{Footnote1}.
\ When $N=2$ this is the familiar Jacobi identity. \ The proof is by direct
calculation and follows as a consequence of associativity.

However, for odd $N$, $\left[  B_{1}\cdots B_{N-1}\left[  B_{N}\cdots
B_{2N-1}\right]  \right]  \neq0$, but instead produces the $\left(
2N-1\right)  $-bracket $\left[  B_{1}\cdots B_{2N-1}\right]  $ upon total
antisymmetrization \cite{deAzcarraga,Curtright}. \ Apparently, the simplest
identity obeyed by odd brackets of only one type, that does \emph{not}
introduce higher-order brackets, requires that they act at least thrice. \ For
any odd $N=2L+1$, a valid relation is the immediate generalization of that
found by Bremner for the case of $3$-brackets. \ To show this, we present two
easily established lemmata, followed by our main theorem and its proof. \ Firstly,

\begin{lemma}%
\begin{equation}
\left[  AB_{1}\cdots B_{J}\right]  =J!\sum_{j=0}^{J}\left(  -1\right)
^{j}B_{1}\cdots B_{j}AB_{j+1}\cdots B_{J}\ . \label{Lemma1}%
\end{equation}

\end{lemma}

\noindent Total antisymmetrization of the $B$s is understood. \ Here we have
also used the convention that an empty product equals 1. \ Explicitly,
$B_{1}\cdots B_{0}=1=B_{J+1}\cdots B_{J}$, so that the first and last terms in
the sum are $\mathcal{A}B_{1}\cdots B_{J}$ and $\left(  -1\right)  ^{J}%
B_{1}\cdots B_{J}\mathcal{A}$, respectively. \ It is a simple exercise to use
this first lemma to prove (\ref{EvenGJI}). \ Similarly,

\begin{lemma}%
\begin{align}
\left[  \mathcal{A}B_{1}\cdots B_{J}\mathcal{Z}\right]   &  =J!\sum_{j=0}%
^{J}\left(  -1\right)  ^{j}\sum_{k=0}^{J-j}\left(  -1\right)  ^{k}B_{1}\cdots
B_{k}\mathcal{A}B_{k+1}\cdots B_{J-j}\mathcal{Z}B_{J-j+1}\cdots B_{J}%
\nonumber\\
&  -J!\sum_{j=0}^{J}\left(  -1\right)  ^{j}\sum_{k=0}^{J-j}\left(  -1\right)
^{k}B_{1}\cdots B_{k}\mathcal{Z}B_{k+1}\cdots B_{J-j}\mathcal{A}%
B_{J-j+1}\cdots B_{J}\ . \label{Lemma2}%
\end{align}

\end{lemma}

\noindent Finally, it is rather tedious but fairly straightforward to use both
lemmata to prove the following.

\begin{theorem}
For associative products, with implicit total antisymmetrization of the $B$s,%
\begin{equation}
\left[  \left[  A\left[  B_{1}\cdots B_{2L+1}\right]  B_{2L+2}\cdots
B_{4L}\right]  B_{4L+1}\cdots B_{6L}\right]  =\left[  \left[  AB_{1}\cdots
B_{2L}\right]  \left[  B_{2L+1}\cdots B_{4L+1}\right]  B_{4L+2}\cdots
B_{6L}\right]  \ . \label{GBI}%
\end{equation}

\end{theorem}

\begin{proof}
[Proof of Theorem]The result (\ref{GBI}) follows from resolving the left- and
right-hand sides into sums of canonically ordered words, as illustrated above
for the case of $3$-brackets. \ We have
\begin{gather}
\left[  \left[  AB_{1}\cdots B_{2L}\right]  \left[  B_{2L+1}\cdots
B_{4L+1}\right]  B_{4L+2}\cdots B_{6L}\right]  =\sum_{n=0}^{6L}\left(
-1\right)  ^{n}m_{n}^{\left(  1\right)  }~B_{1}\cdots B_{n}~A~B_{n+1}\cdots
B_{6L}\ ,\nonumber\\
\left[  \left[  A\left[  B_{1}\cdots B_{2L+1}\right]  B_{2L+2}\cdots
B_{4L}\right]  B_{4L+1}\cdots B_{6L}\right]  =\sum_{n=0}^{6L}\left(
-1\right)  ^{n}m_{n}^{\left(  2\right)  }~B_{1}\cdots B_{n}~A~B_{n+1}\cdots
B_{6L}\ .
\end{gather}
All the coefficients $m_{n}^{\left(  1,2\right)  }$ in these two resolutions
are manifestly positive integers. \ The theorem is established by showing that
$m_{n}^{\left(  1\right)  }=m_{n}^{\left(  2\right)  }$ for all $n$.

By direct calculation, through the use of the two lemmata, we find
\begin{align}
m_{n}^{\left(  1\right)  }  &  =m_{n}^{\left(  2\right)  }=\left(
2L+1\right)  !\left(  2L\right)  !\left(  2L-1\right)  !\times c_{n}%
\ ,\label{Mn}\\
c_{n}  &  =\left\{
\begin{array}
[c]{ll}%
\left(  n+1\right)  \left(  4L-n\right)  /2 & \text{for \ \ }0\leq n\leq2L\\
10L^{2}-6Ln+L+n^{2} & \text{for \ \ }2L+1\leq n\leq3L\\
c_{6L-n} & \text{for \ \ }3L+1\leq n\leq6L
\end{array}
\right.  \ . \label{Cn}%
\end{align}
The determination of the $m_{n}$s is just a matter of enumerating the ways to
obtain a particular intercalation of $A$ among the $B$s. \ 

Consider in more detail some of the calculations involved. \ As a first step,
with the implicit antisymmetrization \cite{Footnote2}, the internal brackets
$\left[  B_{1}\cdots B_{2L+1}\right]  $ or $\left[  B_{2L+1}\cdots
B_{4L+1}\right]  $ may be supplanted by products: \ $\left[  B_{1}\cdots
B_{2L+1}\right]  =\left(  2L+1\right)  !~\left(  B_{1}\cdots B_{2L+1}\right)
$ or $\left[  B_{2L+1}\cdots B_{4L+1}\right]  =\left(  2L+1\right)  !~\left(
B_{2L+1}\cdots B_{4L+1}\right)  $. \ Then we may write, on the one hand,%
\begin{align}
&  \left[  \left[  AB_{1}\cdots B_{2L}\right]  \left[  B_{2L+1}\cdots
B_{4L+1}\right]  B_{4L+2}\cdots B_{6L}\right] \nonumber\\
&  =-\left(  2L+1\right)  !\left[  \left[  AB_{1}\cdots B_{2L}\right]
B_{4L+2}\cdots B_{6L}\left(  B_{2L+1}\cdots B_{4L+1}\right)  \right]
\end{align}
In this expression, we may now rename indices, bearing in mind the
antisymmetrization.%
\begin{align}
&  \left[  \left[  AB_{1}\cdots B_{2L}\right]  \left[  B_{2L+1}\cdots
B_{4L+1}\right]  B_{4L+2}\cdots B_{6L}\right] \nonumber\\
&  =\left(  2L+1\right)  !\left[  \left[  AB_{2L}\cdots B_{4L-1}\right]
B_{1}\cdots B_{2L-1}\left(  B_{4L}\cdots B_{6L}\right)  \right]  \ .
\label{m(1)FirstStep}%
\end{align}
Next, we apply Lemma 2 for $J=2L-1$, and identify $\left[  AB_{2L}\cdots
B_{4L-1}\right]  $ with $\mathcal{A}$, and $\left(  B_{4L}\cdots
B_{6L}\right)  $ with $\mathcal{Z}$. \
\begin{subequations}
\begin{align}
\left[  \mathcal{A}B_{1}\cdots B_{2L-1}\mathcal{Z}\right]   &  =\left(
2L-1\right)  !\sum_{j=0}^{2L-1}\left(  -1\right)  ^{j}\sum_{k=0}%
^{2L-1-j}\left(  -1\right)  ^{k}B_{1}\cdots B_{k}\mathcal{A}B_{k+1}\cdots
B_{2L-1-j}\mathcal{Z}B_{2L-j}\cdots B_{2L-1}\label{ZRight}\\
&  -\left(  2L-1\right)  !\sum_{j=0}^{2L-1}\left(  -1\right)  ^{j}\sum
_{k=0}^{2L-1-j}\left(  -1\right)  ^{k}B_{1}\cdots B_{k}\mathcal{Z}%
B_{k+1}\cdots B_{2L-1-j}\mathcal{A}B_{2L-j}\cdots B_{2L-1}\ . \label{ZLeft}%
\end{align}
\ To continue, consider first the coefficients $m_{n}^{\left(  1\right)  }$
where $n\leq2L$. \ 

For the determination of $m_{n\leq2L}^{\left(  1\right)  }$, since
$\mathcal{Z}$ consists of $\left(  2L+1\right)  $ $B$s, it \emph{must} be
placed to the \emph{right} of $\mathcal{A}$ in the application of Lemma 2.
\ Otherwise there would be too many $B$s to the left of $A$. \ Thus for
$m_{n\leq2L}^{\left(  1\right)  }$ we need keep only the first line in the
last relation, (\ref{ZRight}). \ To place a total of $n$ $B$s to the left of
the $A$ contained in $\mathcal{A=}\left[  AB_{2L}\cdots B_{4L-1}\right]  $,
with $k$ $B$s already to the left as in (\ref{ZRight}), we then need only the
terms in $\mathcal{A}$\ with an additional $\left(  n-k\right)  $ $B$s to the
left of $A$. \ That is to say, from Lemma 1, with $J=2L$ and all $B$ indices
shifted up by $2L-1$,
\end{subequations}
\begin{equation}
\mathcal{A=}\left[  AB_{2L}\cdots B_{4L-1}\right]  =\left(  2L\right)
!\sum_{l=0}^{2L}\left(  -1\right)  ^{n}B_{2L}\cdots B_{l+2L-1}AB_{l+2L}\cdots
B_{4L-1}\ ,
\end{equation}
and from this we need only the term with $l=n-k$. \ The net result for
$m_{n\leq2L}^{\left(  1\right)  }$ is
\begin{equation}
m_{n\leq2L}^{\left(  1\right)  }=\left(  2L+1\right)  !\left(  2L\right)
!\left(  2L-1\right)  !\times c_{n\leq2L}\ ,
\end{equation}%
\begin{equation}
c_{n\leq2L}=\left.  \sum_{j=0}^{2L-1}\sum_{k=0}^{2L-1-j}\sum_{l=0}^{2L}%
\delta_{l,n-k}\right\vert _{n\leq2L}=\sum_{j=0}^{2L-1}~\sum_{k=0}^{\min\left(
n,2L-1-j\right)  }1=\frac{\left(  n+1\right)  \left(  4L-n\right)  }{2}\ .
\label{cn<=2L}%
\end{equation}

On the other hand, with similar steps, we have%
\begin{align}
&  \left[  \left[  A\left[  B_{1}\cdots B_{2L+1}\right]  B_{2L+2}\cdots
B_{4L}\right]  B_{4L+1}\cdots B_{6L}\right] \nonumber\\
&  =\left(  2L+1\right)  !\left[  \left[  AB_{1}\cdots B_{2L-1}\left(
B_{2L}\cdots B_{4L}\right)  \right]  B_{4L+1}\cdots B_{6L}\right]  \ .
\label{m(2)FirstStep}%
\end{align}
We again apply Lemma 2 for $J=2L-1$, but to $\left[  AB_{1}\cdots
B_{2L-1}\left(  B_{2L}\cdots B_{4L}\right)  \right]  $, so now we identify $A$
with $\mathcal{A}$, and $\left(  B_{2L}\cdots B_{4L}\right)  $ with
$\mathcal{Z}$. \ As before, consider first only $m_{n}^{\left(  2\right)  }$
coefficients where $n\leq2L$. \ For the determination of $m_{n\leq2L}^{\left(
2\right)  }$, $\mathcal{Z}$ \emph{must} once again be placed to the
\emph{right} of $\mathcal{A}$, so we need keep only the line (\ref{ZRight}).
\ We pick up an additional $\left(  n-k\right)  $ $B$s by applying again Lemma
1, only this time to the remaining \emph{outside} bracket in
(\ref{m(2)FirstStep}). \ The net result for $m_{n\leq2L}^{\left(  2\right)  }$
is%
\begin{equation}
m_{n\leq2L}^{\left(  2\right)  }=\left(  2L+1\right)  !\left(  2L\right)
!\left(  2L-1\right)  !\times c_{n\leq2L}\ ,
\end{equation}
with exactly the same expression for $c_{n\leq2L}$ as before, (\ref{cn<=2L}).
\ Thus we have shown $m_{n\leq2L}^{\left(  1\right)  }=m_{n\leq2L}^{\left(
2\right)  }$.

Next, consider the coefficients where $2L+1\leq n\leq3L$. \ There are still
contributions to either $m_{n}^{\left(  1\right)  }$ or $m_{n}^{\left(
2\right)  }$ from the line (\ref{ZRight}), as above, of the form $\left(
2L+1\right)  !\left(  2L\right)  !\left(  2L-1\right)  !\times c_{n}$, and
these contributions to either $m_{n}^{\left(  1\right)  }$ or $m_{n}^{\left(
2\right)  }$\ still turn out to be the same. But in this case the sums
contributing to $c_{n}$ give%
\begin{equation}
\left.  \sum_{j=0}^{2L-1}\sum_{k=0}^{2L-1-j}\sum_{l=0}^{2L}\delta
_{l,n-k}\right\vert _{2L+1\leq n\leq3L}=\sum_{j=0}^{2L-1-\left(  n-2L\right)
}\sum_{k=\left(  n-2L\right)  }^{2L-1-j}1=\frac{1}{2}\left(  4L+1-n\right)
\left(  4L-n\right)  \ .
\end{equation}
Moreover, from applying Lemma 2, there are now contributions to either
$m_{n}^{\left(  1\right)  }$ or $m_{n}^{\left(  2\right)  }$ from the second
line, (\ref{ZLeft}), where the respective $\mathcal{Z}$s are placed to the
\emph{left} of the $\mathcal{A}$s. \ Following steps similar to those above,
it is not difficult to see that these other terms contribute the \emph{same
amount} to either $m_{n}^{\left(  1\right)  }$ or $m_{n}^{\left(  2\right)  }%
$, for $2L+1\leq n\leq3L$. \ Namely, $\left(  2L+1\right)  !\left(  2L\right)
!\left(  2L-1\right)  !\times$%
\begin{equation}
\left.  \sum_{j=0}^{2L-1}\sum_{k=0}^{2L-1-j}\sum_{l=0}^{2L}\delta
_{l,n+j-4L}\right\vert _{2L+1\leq n\leq3L}=\sum_{j=4L-n}^{2L-1}\sum
_{k=0}^{2L-1-j}1=\frac{1}{2}\left(  n-2L+1\right)  \left(  n-2L\right)  \ .
\end{equation}
Thus the net result is $m_{2L+1\leq n\leq3L}^{\left(  1\right)  }=m_{2L+1\leq
n\leq3L}^{\left(  2\right)  }=\left(  2L+1\right)  !\left(  2L\right)
!\left(  2L-1\right)  !\times c_{2L+1\leq n\leq3L}$ with%
\begin{equation}
c_{2L+1\leq n\leq3L}=\frac{1}{2}\left(  4L+1-n\right)  \left(  4L-n\right)
+\frac{1}{2}\left(  n-2L+1\right)  \left(  n-2L\right)  =10L^{2}%
-6Ln+L+n^{2}\ .
\end{equation}

Finally, consider the coefficients for $3L+1\leq n\leq6L$. \ These are given
by an elementary reflection symmetry: \ $m_{n}^{\left(  1\right)  }%
=m_{6L-n}^{\left(  1\right)  }$ and $m_{n}^{\left(  2\right)  }=m_{6L-n}%
^{\left(  2\right)  }$. \ Thus $m_{n}^{\left(  1\right)  }=m_{n}^{\left(
2\right)  }=\left(  2L+1\right)  !\left(  2L\right)  !\left(  2L-1\right)
!\times c_{6L-n}$ for $3L+1\leq n\leq6L$.
\end{proof}

As a check, the coefficients must sum to give the number of generic terms that
appear in three nested $\left(  2L+1\right)  $-brackets (i.e. in either
$\left[  \left[  \left[  \cdots\right]  \cdots\right]  \cdots\right]  $ or
$\left[  \left[  \cdots\right]  \cdots\left[  \cdots\right]  \right]  $).
\ That is, $\sum_{n=0}^{6L}m_{n}=\left(  \left(  2L+1\right)  !\right)  ^{3}$.
\ Equivalently,
\begin{equation}
\sum_{n=0}^{6L}c_{n}=2L\left(  2L+1\right)  ^{2}\text{ .} \label{Check}%
\end{equation}
This condition is indeed satisfied by the $c_{n}$ given in (\ref{Cn}).

\section{Conclusion}

Perhaps N-brackets and algebras have an important role to play in physics, as
originally suggested by Nambu. \ Recently there has been considerable interest
in $N$-brackets, especially $3$-brackets, as expressed in the physics
literature (see \cite{CFJMZ}\ and references therein). \ These ideas await
further development.

\begin{acknowledgments}
We thank David Fairlie and Cosmas Zachos for sharing their thoughts about
Nambu brackets. \ We also thank the referee for comments which significantly
improved the exposition of this paper. \ One of us (TC) further thanks the
Lago Mar Resort for providing the beautiful and stimulating surroundings where
portions of this work were completed. \ This work was supported by NSF Awards
0555603 and 0855386.
\end{acknowledgments}

\end{document}